\documentclass[journal]{IEEEtran}
\usepackage{mathrsfs}
\usepackage{amsmath}
\usepackage{amssymb}
\usepackage{amsthm}
\usepackage{stfloats}
\usepackage{cuted}\stripsep -3pt plus 3pt minus 2pt
\ifCLASSINFOpdf \else \fi

\usepackage[dvips]{graphicx}
\usepackage{algorithm}
\usepackage{color}
\usepackage[T1]{fontenc}
\usepackage{booktabs}
\usepackage[normalem]{ulem} 
\usepackage{slashbox}
\usepackage{diagbox}
\usepackage{subfigure}
\usepackage{xcolor}
\usepackage{setspace}
\usepackage{makecell}
\usepackage{algorithm} 
\usepackage{algpseudocode}
\usepackage{caption}
\usepackage{upgreek}

\newtheorem{Theo}{Theorem}

\makeatletter
\renewcommand{\maketag@@@}[1]{\hbox{\m@th\normalsize\normalfont#1}}%
\makeatother

\begin{document}
\title{A 5G DMRS-based Signal for Integrated Sensing and Communication System}
\author{Zhiqing Wei,~\IEEEmembership{Member,~IEEE},
		Fengyun Li, 
		Haotian Liu,
		Xu Chen,~\IEEEmembership{Student Member,~IEEE},\\
		Huici Wu,~\IEEEmembership{Member,~IEEE}, Kaifeng Han,~\IEEEmembership{Member,~IEEE}
		Zhiyong Feng,~\IEEEmembership{Senior Member,~IEEE}

\thanks{
	
	Zhiqing Wei, Fengyun Li, Haotian Liu, Xu Chen, Huici Wu, and Zhiyong Feng are with Beijing University of Posts and Telecommunications, Beijing 100876, China (emails: weizhiqing@bupt.edu.cn,  lfy@bupt.edu.cn, 15738993901@163.com, chenxu96330@bupt.edu.cn, dailywu@bupt.edu.cn, fengzy@bupt.edu.cn).

Kaifeng Han is with China Academy of Information and Communications Technology,
Beijing 100191, China.
(email: hankaifeng@caict.ac.cn).}}

\maketitle

\begin{abstract}

Integrated sensing and communication (ISAC) system can not only make full use of precious spectrum resources and simplify hardware equipment, but also combine the advantages of both to make communication and sensing, thereby obtaining greater performance gain. This paper focuses on the design of integrated sensing and communication signals in the Internet of Vehicles scenario, based on  5G NR uses demodulation reference signal (DMRS) to achieve sensing without changing the original communication waveform. Communication signals and sensing signals do not interfere with each other. The root mean square error of ranging and velocity are obtained at 1.2603 m and 0.6199 m/s. The sensing performance was analyzed from the two perspectives of simulation results and theoretical derivation, and the influence of center frequency, number of IFFT/ FFT points, subcarrier spacing on the accuracy of ranging and speed measurement were summarized. 
\end{abstract}
\IEEEpeerreviewmaketitle

\section{Introduction}
\label{sec1}

 With the official approval of a series of new standards for Release 18 (Rel-18) at the "3rd Generation Partnership Project" plenary meeting (3GPP RAN) at the end of 2021, the standardization of the fifth generation(5G) mobile systems has officially entered the "second half", namely 5G-advanced (5G-A) stage[1]. 5G-A and the sixth generation (6G) mobile communication systems will support services such as Internet of vehicles and extended reality (XR). With the intelligent transportation system (ITS) as an example, the car can detect the surrounding driving environment and exchange information such as vehicle velocity, acceleration, braking, and driving direction with base stations to ensure road safety and health of people. On the other hand, it needs to support the fusion and sharing of sensing information. This requires the vehicle to have both high-quality sensing and communication capabilities. ISAC technology is an inevitable trend in the development of the Internet of vehicles. Waveform design is the first task in realizing ISAC technology.

There are three types of integrated waveform design directions:(1) The integrated waveform is based on radar waveform, on which the communication data is modulated, such as linear frequency modulated (LFM) [2] and frequency modulated continuous wave (FMCW) [3]. Because the communication rate of low-order modulation is not high, [4] uses the LFM signal as the carrier of the quadrature amplitude modulation (QAM) signal. The ambiguity function of the MQAM-LFM signal is close to the pushpin type, but considering the increase of signal amplitude types, the dynamic range of the signal becomes larger, and the linearity of the RF front-end amplifier is required to be higher.
(2) The new integrated waveform design is designed according to the objectives and constraints of radar and communication in ISAC system [5]. [6] optimizes the precoding matrices of radar and communication waveforms respectively, the minimum signal-to-interference-noise ratio (SINR) of the communication user end is constrained and the transmission is minimized. The dual-function waveform is not limited by conventional waveforms and can achieve a compromise between radar and communication performance. However, there are still certain challenges in terms of computational complexity and hardware cost. For example, the generation of dedicated waveforms involves the solution of complex optimization problems and depends on the state information of communication channels, which is not easy to obtain in high-velocity mobile scenarios [7].
(3) The integrated waveform is dominated by communication waveforms, generally orthogonal frequency division multiplexing (OFDM) signals. [8] compensates and decoherents the echo communication information, and uses subspace projection to achieve joint high-resolution estimation of range and velocity. The random communication information on different subcarriers will affect the fuzzy function and reduce the sensing performance. [9] performs weighted preprocessing on the integrated signal, uses the whale optimization algorithm to set the weighting coefficient to optimize its fuzzy function. Pilot frequency, synchronization and cyclic prefix will lead to the appearance of spurious peaks in the point spread function of radar matched filtering, which seriously affects the sensing performance. [10] proposed a method to push the point spread function side lobes and spurious peaks of OFDM signals outside the radar observation window. The disadvantage of high peak-to-average power ratio (PAPR) of OFDM waveforms results in low power efficiency in radar applications. [11] studied the control of the maximum PAPR of the waveform by weighting. 

Sensing requires signals with structured and strong autocorrelation characteristics, while communication requires random signals to maximize the data rate [12]. There is a contradiction between the two. The above integrated signals all need to transform the original signal, and there are various limitations. Reference signal in 5G NR has good passive detection performance and strong anti-interference ability [13] , and its sensing ability has attracted extensive attention. We proposes a signal design method to achieve integrated sensing and communication based on the 5G NR standard, and expands the sensing function of the reference signal in 5G NR. The contributions are summarized as follows.

%%------------contribution----------------------
\begin{itemize}
	{\item Based on DMRS and reference signals, a system structure for ISAC is established, and a model for transmitting and receiving signals is derived.}
	
	{\item Study the radar speed measurement and ranging algorithm, and simulate and verify the feasibility of the designed signal.}
	
	{\item The sensing performance is derived from a theoretical point of view. The specific evaluation criteria include resolution, maximum sensing range and Cramero boundary. The factors affecting sensing error are quantitatively discussed to verify the qualitative research conclusion.}
\end{itemize}

The structure of this paper is as follows. 
Section \uppercase\expandafter{\romannumeral2} proposes the ISAC signal model based on DMRS and designs the signal processing scheme. 
Section \uppercase\expandafter{\romannumeral3} analyzes and derives the 
resolutions, the maximum detection distance/velocity and CRLBs of the proposed ISAC signal.  
In Section \uppercase\expandafter{\romannumeral4}, Matlab simulation was used to implement radar detection.
Simulation results also reveal the influence of multiple parameters such as carrier frequency and 
subcarrier spacing on sensing accuracy. 
Section \uppercase\expandafter{\romannumeral5} concludes this paper and discusses the future work.

\section{ISAC Signal Model}
\label{sec2}

In this section, a system based on mathematical model and system structure of ISAC signal were developed, and the signal processing scheme of ISAC signal was designed to facilitate radar signal processing.

\subsection{Time-frequency Domain Mapping Scheme of DMRS}

The mapping type of DMRS determines the starting position of DMRS symbols in the time domain. In order to occupy rich time domain resources, this article chooses mapping type A, where the first DMRS symbol is located at symbol \#2. After selecting the time domain mapping type, continue to configure frequency domain resources for DMRS. In order to occupy more frequency domain resources, this article chooses type 1. DMRS RE is distributed in the frequency domain interval of a certain symbol with a density of 50 \%.

The DMRS structure of 5G NR is composed of a front-end DMRS and an additional DMRS with configurable time domain density. The pattern of the additional DMRS is the same as the front-end DMRS. The pre-DMRS must exist in each scheduling time unit. In order to reduce the delay of demodulation and d pre-DMRS ecoding, the position where the DMRS first appears should be as close as possible to the starting point of scheduling. In medium and high-velocity scenarios, more DMRS symbols need to be inserted within the scheduling duration to meet the estimation accuracy of the time-varying channel, that is, additional DMRS. This article chooses the mode of single symbol plus 3 additional DMRS, as shown in Fig. \ref{fig_DMRS_time_frequency_domain}, in order to have the richest time-frequency domain resources and obtain optimal sensing performance.

\begin{figure}[!t]
	\centering
  	\includegraphics[width=0.5\textwidth]{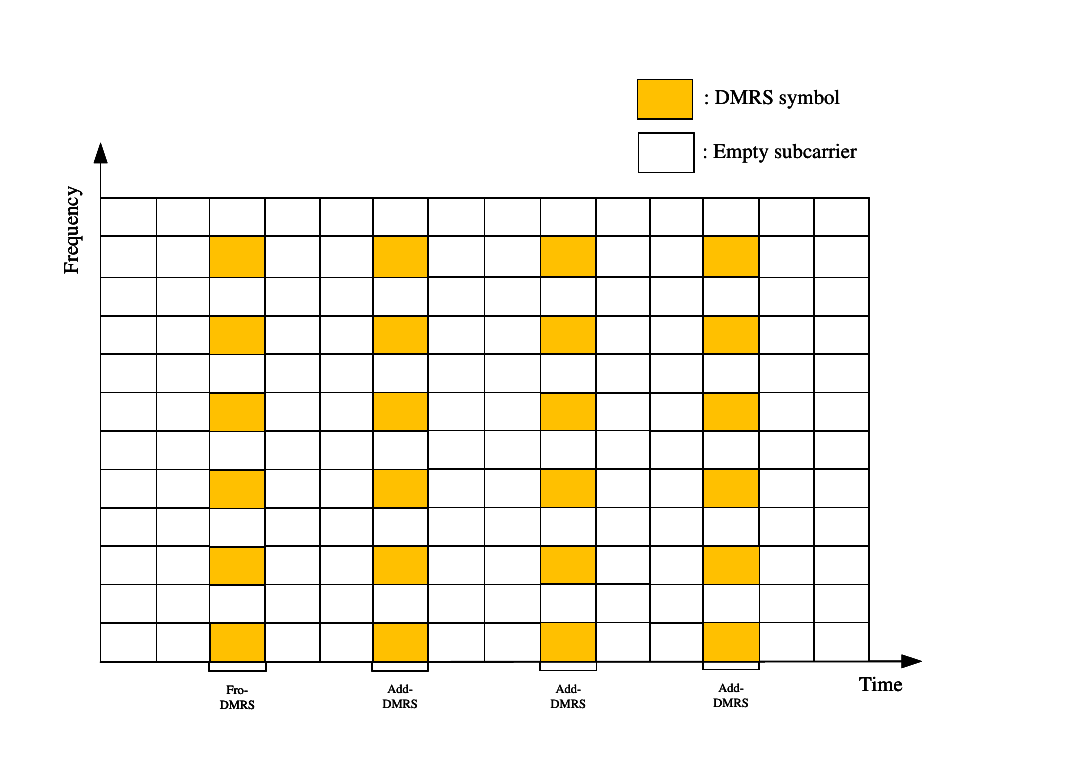}
	\caption{Single DMRS symbol plus three additional DMRS symbols} \label{fig_DMRS_time_frequency_domain}
\end{figure}

\subsection{The Transmitted and Received Signal Model}

The DMRS occupying $M_J$ OFDM symbols and $N_J$ subcarriers is expressed as 
\begin{equation}
\label{eq_5}\small	
x(t)=\sum_{m=0}^{M_{J}-1}\sum_{k=0}^{N_{J}-1}\mathbf{s}(k,m)e^{j2\pi f_{k}t} {\rm {rect}}(\frac{t-mT_{\rm s}}{T_{\rm s}}),
\end{equation}
where $\mathbf{s}$ is the modulated reference signal symbol, 
$k$ is the subcarrier index, $m$ is the OFDM symbol index, $T_{\rm s}=T+T_{\rm{cp}}$ is composed of the OFDM symbol duration $T$ and the cyclic prefix duration $T_{\rm{cp}}$,
$f_k$ is the $k$-th subcarrier carrying the reference signal, 
and ${\rm {rect}} (t/{T_{\rm s}})$ is the rectangle window function. Due to the superposition of multiple subcarriers, the signal has an approximately constant power spectral density, which makes the signal appear to have noise-like characteristics, and has good low interception characteristics for radar detection.

The transmitted signal is reflected by the target vehicle and propagated in free space, and the target echo signal including the two-way round-trip delay and Doppler frequency shift can be expressed as
\begin{equation}
\label{eq_6}\small
\begin{split}
\begin{aligned}
d_{\rm{rx}}(t)=&\sum_{m=0}^{M_{J}-1}\sum_{k=0}^{N_{J}-1} \xi \mathbf{S}_{\rm{tx}}(k,m)e^{j2\pi f_{k}(t- \tau)}e^{j2\pi f_{d,r}t} \\   &{\rm {rect}}(\frac{t-mT_{\rm s}- \tau}{T_{\rm s}}),
\end{aligned}
\end{split}
\end{equation}
where $\xi$ is the attenuation factor during the transmission process and is constant, $\mathbf{S}_{\rm{tx}}$ is is the modulation symbol of transmission, 
$c$ is the velocity of light,
$f_{d,r}$ is the Doppler frequency shift,
and $\tau=\frac{2R}{c}$ is delay with $R$ representing the range of target.

Simplifying (\ref{eq_6}), we have
\begin{equation}
\label{eq_7}\small	
d_{\rm{rx}}(t)=\sum_{m=0}^{M_{J}-1}\sum_{k=0}^{N_{J}-1}  \mathbf{S}_{\rm{rx}}(k,m)e^{j2\pi f_{k}t} {\rm {rect}}(\frac{t-mT_{\rm s}- \tau}{T_{\rm s}}),
\end{equation}
where $\mathbf{S}_{\rm{rx}}$ is the received modulation symbol. Each transmitted modulation symbol will have different degrees of fading in the integrated channel of sensing and communication. $\mathbf{S}_{\rm{rx}}(k,m)$ can be expressed as
\begin{equation}
\begin{split}
\begin{aligned}
\label{eq_9}\small
\mathbf{S}_{\rm{rx}}(k,m)&=\xi \mathbf{S}_{\rm{tx}}(k, m) e^{j2 \pi K_{\rm{comb}}^{\rm{symbol}}mT_{s}f_{d,r}}e^{-j2 \pi K_{\rm{comb}}^{\rm{carrier}}k \Delta f \frac{2R}{c}},
\end{aligned}
\end{split}
\end{equation}
where $K_{\rm{comb}}^{\rm{symbol}}$ is the period of symbol interval, 
$K_{\rm{comb}}^{\rm{carrier}}$ is the period of subcarrier interval, 
$\Delta f=1/T$ is the subcarrier spacing, $N \times 1$ 
dimensional vector $\mathbf{k}_{d}$ is the phase shift caused by delay, which represents the distance information of target, 
$1 \times M$ dimensional vector $\mathbf{k}_{v}$ is the phase shift caused by Doppler shift, which represents the velocity information of target, 
and $\otimes$ refers to Kronecker product.

The above equation shows that in the case of a fixed subcarrier, the Doppler frequency shift caused by the velocity change will only introduce a phase shift between different modulation symbols on the same subcarrier, that is, only affect the time domain. In the case of a fixed OFDM symbol, the time delay caused by the range change will only cause a phase shift between different subcarriers of the same symbol, that is, only affect the frequency domain. Based on this, the time domain and frequency domain can be processed separately to obtain time delay and Doppler frequency shift information respectively, and thus obtain range and velocity information.

The ISAC signal processing is similar to the OFDM system process. The difference is that the DMRS sequence needs to be embedded at the transmitter. That is, before digital modulation, the DMRS composed of the Gold sequence is mapped to the corresponding OFDM symbol according to the time-frequency domain mapping scheme. And then digital modulation, after target reflection and free space propagation in the channel, the receiving end performs operations on the echo signal. After parallel-to-serial conversion, a suitable radar signal processing algorithm is used to extract the distance and speed information of the target.

\section{Sensing Performance Analysis}
\label{sec3}

In this section, we introduces the signal processing process of DMRS using the sturm algorithm, and provides an operation method for processing ISAC signals for subsequent theoretical derivation.

\subsection{Resolution of Distance and Velocity Estimation}

Performing IFFT on the $m$-th column of $\mathbf{S}_g$, we have
\begin{equation}
\label{eq_14}\small
{\rm{IFFT}}(\mathbf{S}_{g,m}(k))=\sum_{k=0}^{N-1}e^{-j2 \pi K_{\rm{comb}}^{\rm{carrier}} k \Delta f \frac{2R}{c}} \times e^{\frac{j2 \pi rk}{N}}.
\end{equation}
Peaks occur when the exponents cancel each other, and the estimated range is written as $\hat{R}$, the number of columns at the peak is recorded as  $ind_{s_{g},m}$, $\hat{R}$ is 
\begin{equation}
\label{eq_15}\small
\hat{R}=\frac{c \hat{r}}{2N K_{\rm{comb}}^{\rm{carrier}} \Delta f}=\frac{ind_{s_{g},m}c}{2N K_{\rm{comb}}^{\rm{carrier}} \Delta f}, \quad ind_{s_{g},m} \in\{0,1, \ldots, N-1\}.
\end{equation}

The maximum detection distance is
\begin{equation}
\label{eq_16}\small
R_{max}=\frac{c}{2 K_{\rm{comb}}^{\rm{carrier}} \Delta f}.
\end{equation} 
where $K_{\rm{comb}}^{\rm{carrier}} \in \{ 2,4\} $ is the comb size of DMRS in frequency domain. 

The distance resolution is
\begin{equation}
\label{eq_17}\small
\Delta R=\frac{c}{2N K_{\rm{comb}}^{\rm{carrier}} \Delta f}.
\end{equation}

It can be known from the above derivation that the maximum detection range can be increased by decreasing $\Delta f$ and the range resolution can be improved by increasing the number of subcarriers and $\Delta f$. $f_c$ does not affect ranging performance. There is a trade-off between maximum detection range and range resolution when choosing the subcarrier interval.

Performing FFT on the $k$-th row of $\mathbf{S}_g$, we have
\begin{equation}
\label{eq_18}\small
{\rm{FFT}} (\mathbf{S}_{g,k}(m))=\sum_{m=0}^{M-1}e^{j2 \pi K_{\rm{comb}}^{\rm{symbol}} mT_{s}f_{d,r}} \times e^{-\frac{j2 \pi dm}{M}}.
\end{equation}
Peaks occur when the exponents cancel each other, and the estimated Doppler shift is recorded as $\hat{f}_{d,r}$, the number of lines at the peak is recorded as $ind_{s_{g},k}$, $\hat{f}_{d,r}$ is 
\begin{equation}
\label{eq_19}\small
\hat{f}_{d,r}=\frac{ind_{s_{g},k}}{MK_{\rm{comb}}^{\rm{symbol}}T_s}.
\end{equation}

The estimated velocity is 
\begin{equation}
\label{eq_21}\small
\hat{v}=\frac{ind_{s_{g},k} c}{2M K_{\rm{comb}}^{\rm{symbol}}T_s f_c}, \quad ind_{s_{g},k} \in\{0,1, \ldots, M-1\}.
\end{equation}

The maximum detection velocity is 
\begin{equation}
\label{eq_22}\small
v_{max}=\frac{c}{2 K_{\rm{comb}}^{\rm{symbol}} T_s f_c}.
\end{equation}
where $K_{\rm{comb}}^{\rm{symbol}} \in \{ 3,4,6,12\} $  is the comb size of DMRS in time domain. 
The velocity resolution is 
\begin{equation}
\label{eq_23}\small
\Delta v=\frac{c}{2MK_{\rm{comb}}^{\rm{symbol}} T_s f_c}.
\end{equation}

From the above derivation, it can be known that the maximum estimation velocity can be increased by reducing $T_s$ and increasing $f_c$. The maximum estimation velocity of the ISAC signal operating at 5.9 GHz is higher than that of the one operating at 24 GHz. The velocity resolution can be improved by increasing the number of symbols, $T_s$ and $f_c$. The velocity resolution of the reference signal operating at 24 GHz is higher than that of the one operating at 5.9 GHz. There is a trade-off between the maximum estimated velocity and velocity resolution when choosing the symbol time and center frequency.

\subsection{CRLB of Distance and Velocity Estimation}

\begin{Theo}\label{th_e_f}
	The CRLBs of DMRS for distance 
	and velocity estimation are 
\begin{equation}
	\label{eq_CRLB_R_DMRS}\footnotesize
	CRLB\left({R} \right){\rm{ = }}\frac{{{c^2}}}{{{\xi ^2} \cdot \gamma  \cdot {{({\rm{2}}\pi \Delta f)}^2} \cdot K_{{\rm{comb}}}^{{\rm{symbol}}}}} \cdot \frac{{{\rm{12}}}}{{{M_J}N({N_J} - {\rm{1}})({\rm{7}}{N_J} + {\rm{1}})}},
\end{equation}
\begin{equation}
	\label{eq_CRLB_V_DMRS}\footnotesize
	CRLB\left({v} \right){\rm{ = }}\frac{{{c^2}}}{{{\xi ^2} \cdot \gamma  \cdot {{({\rm{2}}\pi {T_s})}^2}{f_c} \cdot K_{{\rm{comb}}}^{{\rm{carrier}}}}} \cdot \frac{{{\rm{12}}}}{{{N_J}M({M_J} - {\rm{1}})({\rm{7}}{M_J} + {\rm{1}})}}.
\end{equation}
	with $\gamma=\frac{A^2}{\sigma ^2}$ is the signal-to-noise ratio (SNR).
\end{Theo}

\begin{proof}
The received signal is 
\begin{equation}
\label{eq_24}\small
z_{m,n}=\xi A_{m,n} e^{j2 \pi K_{\rm{comb}}^{\rm{symbol}}mT_{s}f_{d,r}} e^{-j2 \pi K_{\rm{comb}}^{\rm{carrier}}n \Delta f \tau} +w_{m,n},
\end{equation}
where $\xi$ is the attenuation factor in the transmission process,  
$w_{m,n}$ is the Gaussian additive white noise obeying  $N(0,\sigma ^2)$,
and $A_{m,n}=\left| x_{m,n}\right|$ is the amplitude of modulation symbol. 
According to Fig. \ref{fig_DMRS_time_frequency_domain}, 
$K_{\rm{comb}}^{\rm{symbol}}$ in this paper is not fixed. 

$z_{m,n}$ with unknown parameters $\theta=(\tau,f_{d,r})$ is observed. 
An estimation of $\theta$ is performed, whose likelihood function is 
\begin{equation}
	\label{eq_25}\small
	\begin{array}{l}
		f\left( {\left. z \right|\tau ,{f_{d,r}}} \right) = \frac{1}{{{{\left( {2\pi {\sigma ^2}} \right)}^{\frac{{MN}}{2}}}}} \cdot\\
		{e^{ - \frac{1}{{2{\sigma ^2}}}\sum\limits_m {\sum\limits_n {{{\left| {{z_{m,n}} - \xi {A_{m,n}}{e^{j2\pi K_{{\rm{comb}}}^{{\rm{symbol}}}m{T_s}{f_{d,r}}}}{e^{ - j2\pi K_{\rm{comb}}^{\rm{carrier}}n \Delta f\tau }}} \right|}^2}} } }}.
	\end{array}
\end{equation}
The log-likelihood function is
\begin{equation}
\label{eq_26}\small
\begin{array}{l}
		L\left( {\left. z \right|\tau ,{f_{d,r}}} \right){\rm{ = }}\ln f\left( {\left. z \right|\tau ,{f_{d,r}}} \right) =  - \frac{{MN}}{2}\ln \left( {2\pi {\sigma ^2}} \right)\\
		- \frac{1}{{2{\sigma ^2}}}\sum\limits_m {\sum\limits_n {{{\left| {{z_{m,n}} - \xi {A_{m,n}}{e^{j2\pi K_{{\rm{comb}}}^{{\rm{symbol}}}m{T_s}{f_{d,r}}}}{e^{ - j2\pi K_{\rm{comb}}^{\rm{carrier}}n\Delta f\tau }}} \right|}^2}} }.
\end{array}
\end{equation}
With $s_{m,n}=\xi {A_{m,n}}{e^{j2\pi K_{{\rm{comb}}}^{{\rm{symbol}}}m{T_s}{f_{d,r}}}}{e^{ - j2\pi K_{\rm{comb}}^{\rm{carrier}}n\Delta f\tau }} $, 
(\ref{eq_26}) can be simplified as	
\begin{equation}
\label{eq_27}\small
L= - \frac{{MN}}{2}\ln \left( {2\pi {\sigma ^2}} \right) - \frac{1}{{2{\sigma ^2}}}\sum\limits_m {\sum\limits_n {{{\left| {{z_{m,n}} - {s_{m,n}}} \right|}^2}} }.
\end{equation}

The Fisher information matrix $\mathbf{F}$ is
\begin{equation}
\begin{aligned}
\label{eq_28}\small
\begin{array}{l}
		\mathbf{F}{\rm{ = }}\left[ {\begin{array}{*{20}{c}}
				{{\mathbf{F}_{\tau \tau }}}&{{\mathbf{F}_{\tau {f_{d,r}}}}}\\
				{{\mathbf{F}_{{f_{d,r}}\tau }}}&{{\mathbf{F}_{{f_{d,r}}{f_{d,r}}}}}
		\end{array}} \right]\\
	=  - \left[ {\begin{array}{*{20}{c}}
				{E\left( {\frac{{{\partial ^2}L}}{{\partial {\tau ^2}}}} \right)}&{E\left( {\frac{{{\partial ^2}L}}{{\partial \tau \partial {f_{d,r}}}}} \right)}\\
				{E\left( {\frac{{{\partial ^2}L}}{{\partial {f_{d,r}}\partial \tau }}} \right)}&{E\left( {\frac{{{\partial ^2}L}}{{\partial {f_{d,r}}^2}}} \right)}
		\end{array}} \right].
\end{array}
\end{aligned}
\end{equation}

Since CRLB matrix for time delay and Doppler estimates is the inverse of the Fisher information matrix, the CRLB matrices of delay and Doppler frequency shift estimation are
\begin{equation}
\label{eq_29}\small
\left[ {\begin{array}{*{20}{c}}
			{CRLB\left( \tau  \right)}&{CRLB\left( {\tau ,{f_{d,r}}} \right)}\\
			{CRLB\left( {{f_{d,r}},\tau } \right)}&{CRLB\left( {{f_{d,r}}} \right)}
\end{array}} \right]{\rm{ = }}{\mathbf{F}^{ - 1}}.
\end{equation}
Thus, we have
\begin{equation}
\label{eq_30}\small
	{\begin{array}{*{20}{l}}
		{CRLB\left( \tau  \right){\rm{ = }}\frac{{{\mathbf{F}_{{f_{d,r}}{f_{d,r}}}}}}{{{\mathbf{F}_{\tau \tau }}{\mathbf{F}_{{f_{d,r}}{f_{d,r}}}} - {\mathbf{F}_{\tau {f_{d,r}}}}{\mathbf{F}_{{f_{d,r}}\tau }}}}},
\end{array}}
\end{equation}

\begin{equation}
\label{eq_32}\small
\begin{array}{*{20}{l}}
	{CRLB\left( {{f_{d,r}}} \right){\rm{ = }}\frac{{{\mathbf{F}_{\tau \tau }}}}{{{\mathbf{F}_{\tau \tau }}{\mathbf{F}_{{f_{d,r}}{f_{d,r}}}} - {\mathbf{F}_{\tau {f_{d,r}}}}{\mathbf{F}_{{f_{d,r}}\tau }}}}}.
\end{array}
\end{equation}

As CRLB is minimum variance of the unbiased estimator, 
mathematics nature of CRLB is the same as that of variance. 
Therefore, the CRLBs for distance and velocity estimation are

\begin{equation}
\begin{split}
\label{eq_33}\small
	CRLB(R) = \frac{{{c^2}}}{4}CRLB(\tau),
\end{split}
\end{equation}
\begin{equation}
\begin{split}
\label{eq_34}\small
	CRLB(v)=\frac{{{c^2}}}{{4f_c^2}}CRLB({{f_{d,r}}}).
\end{split}
\end{equation}
\end{proof}

It can be seen that the longer $T_s$, the smaller the CRLB for velocity estimation and the larger the CRLB for ranging. For a specific application environment, a compromise between velocity estimation and ranging needs to be considered.

\section{Simulation Results and Analysis}
\label{sec4}
The sensing performance of DMRS is verified in this section. 
2D-FFT algorithm \cite{TR:CR21} is employed in distance and velocity estimation. 
The simulation results verify the feasibility of this solution and study the impact of different simulation parameters on sensing performance. The parameters in simulation are shown in Table II. 
\begin{table*}[htbp]
	\label{tab_Simulation parameters}
	\renewcommand\arraystretch{1.2}
	\caption{Simulation parameters}
	\begin{center}
		\begin{tabular}{|c|c|c|}
			\hline
			Symbol&Parameter&Value\\
			\hline
			{$T_{\rm{sample}}$}&\makecell[l]{Sampling interval \quad}&\makecell[l]{$0.016\ {\rm{\upmu s}}$}\\
			\hline
			{$T_{\rm{s}}$}&\makecell[l]{Total OFDM symbol duration \quad}&\makecell[l]{$8.92\ {\rm{\upmu s}}$}\\
			\hline
			\makecell[c]{$T_{\rm{cp}}$}&\makecell[l]{Duration of CP \quad}&\makecell[l]{$0.57\ {\rm{\upmu s}}$}\\
			\hline
			\makecell[c]{$T$}&\makecell[l]{Duration of OFDM symbol \quad}&\makecell[l]{$8.33\ {\rm{\upmu s}}$}\\
			\hline
			\makecell[c]{$\Delta f$}&\makecell[l]{Subcarrier spacing \quad}&\makecell[l]{120\ kHz}\\
			\hline
			\makecell[c]{$f_{\rm{c}}$}&\makecell[l]{Center frequency \quad}&\makecell[l]{24\ GHz}\\
			\hline
			\makecell[c]{$\gamma$}&\makecell[l]{Signal to noise ratio \quad}&\makecell[l]{$\left[-15,10\right]\ \rm{dB}$}\\
			\hline
			\makecell[c]{$N$}&\makecell[l]{Number of subcarriers \quad}&\makecell[l]{256}\\
			\hline
			\makecell[c]{$M$}&\makecell[l]{Number of symbols \quad}&\makecell[l]{140/28}\\
			\hline
		\end{tabular}
	\end{center}
\end{table*}

\subsection{Distance Estimation}

Based on the proposed ISAC signal, detecting a target of a relative range of 48 m with 512 subcarriers is achieved. Referring to the analysis in Section IV, the estimated range $\hat{R}=\frac{ind_{s_{g},m}c}{2N \Delta f}=48.83$ m.

Fig. \ref{fig_distance48-RMSE-CRLB} illustrates the comparison of RMSE and the root CRLB for evaluating distances when utilizing PRS, DMRS, and data signals at varying SNR. Compared to using PRS or data signal, DMRS notably increases the precision of distance measurements at lower SNR thresholds. Observations indicate that the RMSE values for distance calculations using PRS, DMRS, and data signals eventually converge at a common figure. The underlying cause is attributable to the peak index of the IFFT, which is restricted to integer values, thus causing the estimated distances derived from the three ISAC signals to align with the identical gridded values. Moreover, as SNR increases, the RMSE linked with the DMRS approach gets closer to the root CRLB.

As Fig. \ref{fig_distance48-RMSE-SCS} shown, a decrease in subcarrier spacing contributes to a reduced error in distance estimation. This trend is rationalized by the fact that the total bandwidth enlarges in response to enhanced subcarrier spacing, provided the subcarrier numbers remains constant.

\begin{figure}
	\centering
	\includegraphics[width=3.4in,height=3in]{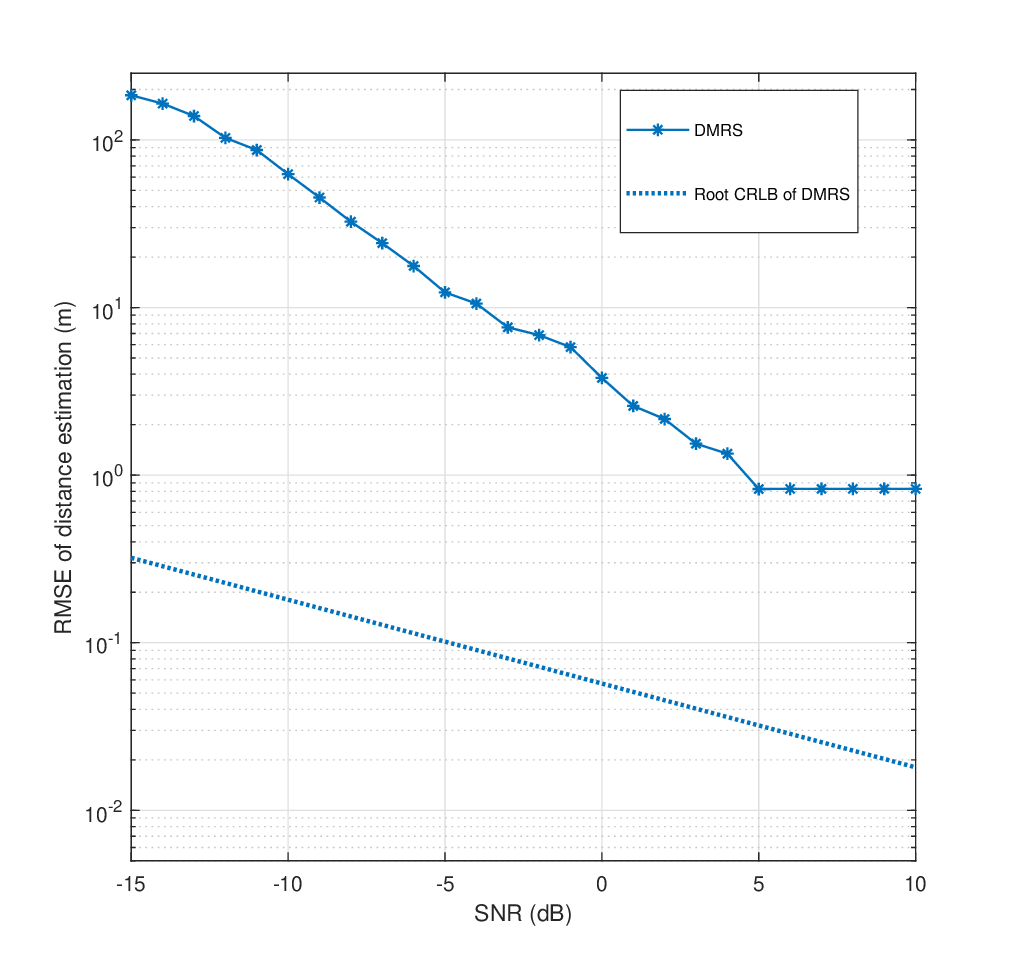}
	\caption{Comparison of RMSEs for distance estimations among PRS, DMRS, and data signal} 
	\label{fig_distance48-RMSE-CRLB}
\end{figure}

\begin{figure}
	\centering
	\includegraphics[width=3.4in,height=3in]{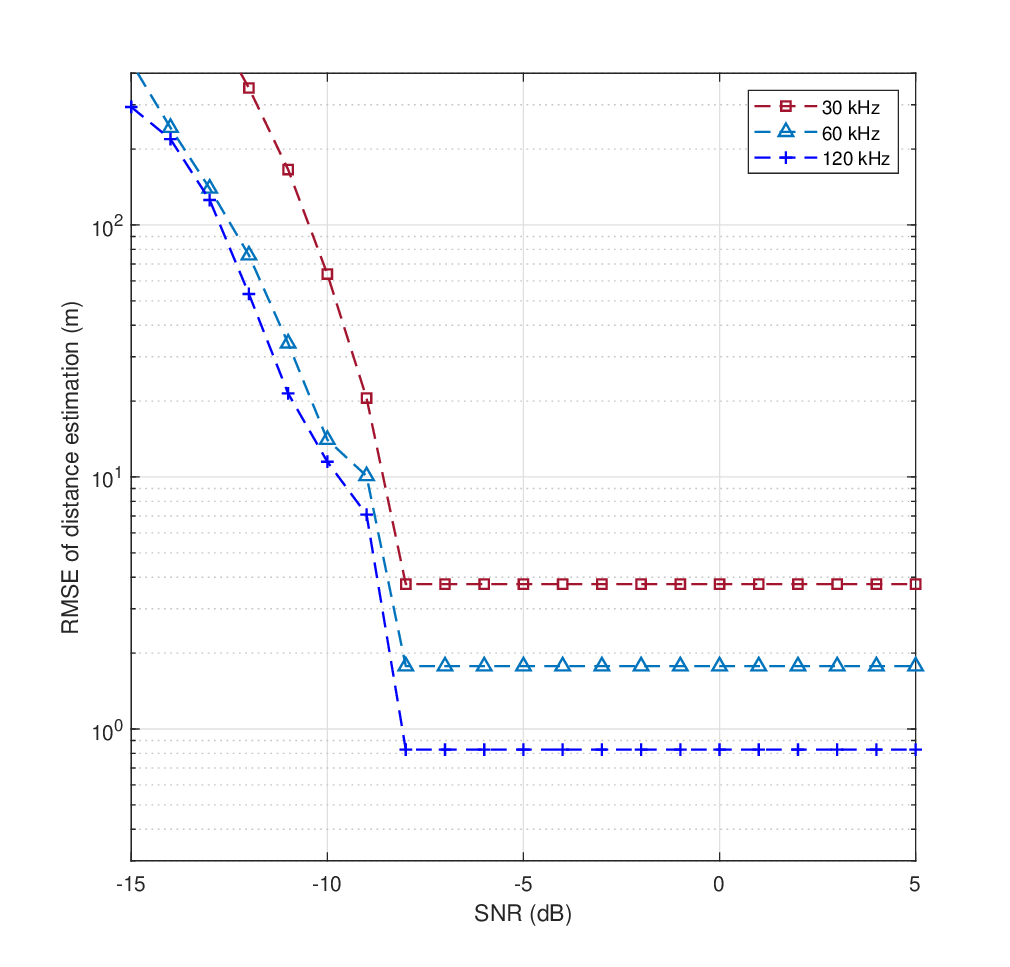}
	\caption{RMSE for distance estimation under different subcarrier spacings} \label{fig_distance48-RMSE-SCS}
\end{figure}

\subsection{Velocity Estimation}

Based on the proposed ISAC signal, detecting a target of a relative velocity of 18 m/s with 140 symbols is realized. Fig. \ref{fig_velocity18-test.eps} shows the echo signal spectrum diagram for velocity estimation, where $\rm index_v$ is the index when searching FFT peak. RMSEs of the single reference signal and 
data signal for velocity estimation are shown in 
Fig. \ref{fig_velocity18-RMSE-CRLB}. 
Due to the limitation of the minimum received SNR for radar detection, 
random and wrong detection result is obtained when the SNR is smaller than $-$10 dB, 
making it meaningless to compare accuracy for 
velocity estimation of different signals. 
With the increase of SNR, the RMSE of DMRS is approaching to the root CRLB. The gap between the RMSE and root CRLB curves for velocity estimation is also caused by the 2D-FFT estimation algorithm.

Fig. \ref{fig_CRLB-T} shows that with the increase of signal duration,
the root CRLB for velocity estimation increases, 
aligning with the trend of RMSE. The reason is that 
the minimum identifiable 
velocity unit by the system is reduced, leading to more precise velocity estimation.

Fig. \ref{fig_CRLB-rv} illustrates the root CRLB for 
distance and velocity estimation with DMRS 
under different subcarrier spacing. 
It is revealed that when the subcarrier spacing is increasing, 
the accuracy of distance estimation is enhancing, 
while the accuracy of velocity estimation is reducing. 
By analyzing 
Fig. \ref{fig_CRLB-T} and the inversely proportional 
relationship between the subcarrier spacing and the signal duration, 
Fig. \ref{fig_CRLB-rv} explains the balance between distance and velocity estimation accuracy.
It's noted that the root CRLBs for both velocity and distance estimations decline as SNR improves.

\begin{figure}[!t]
	\centering
	\includegraphics[width=3.4in,height=3in]{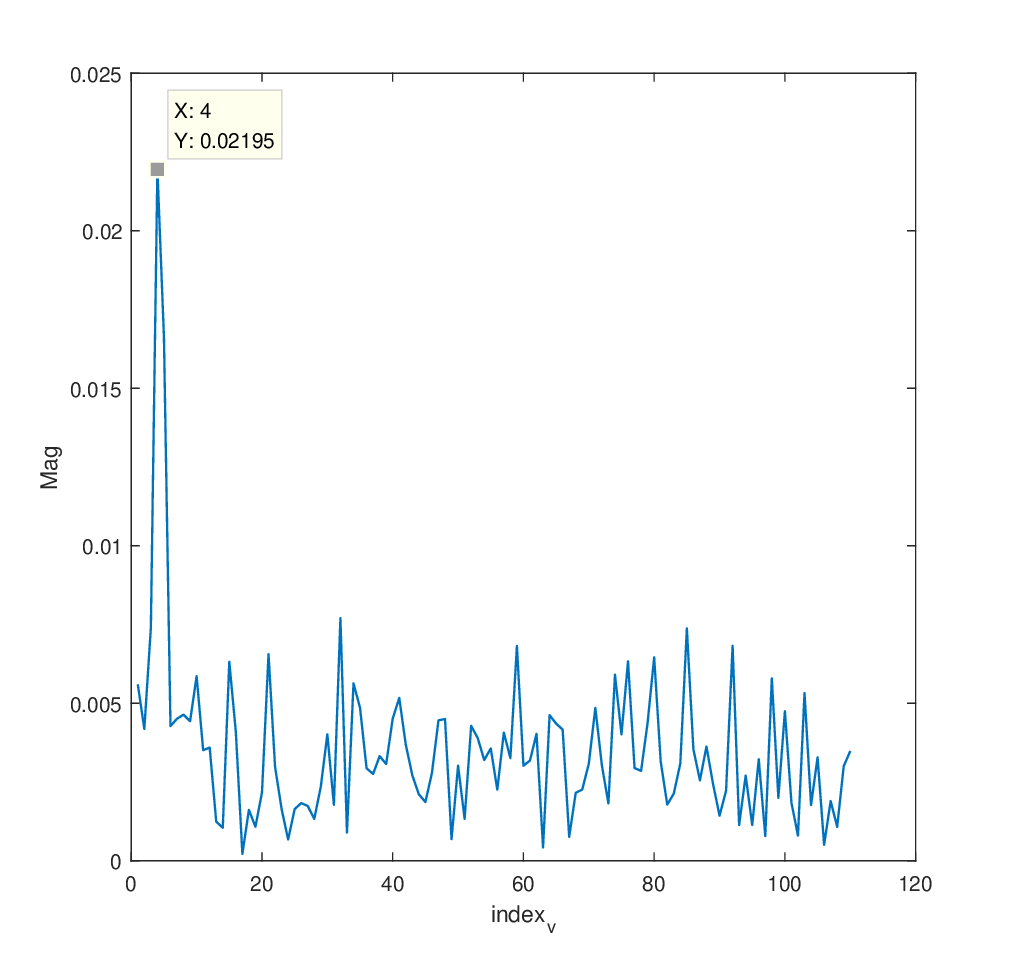}
	\caption{The velocity estimation of the proposed ISAC signal} 
	\label{fig_velocity18-test.eps}
\end{figure}

\begin{figure}[!t]
	\centering
	\includegraphics[width=3.4in,height=3in]{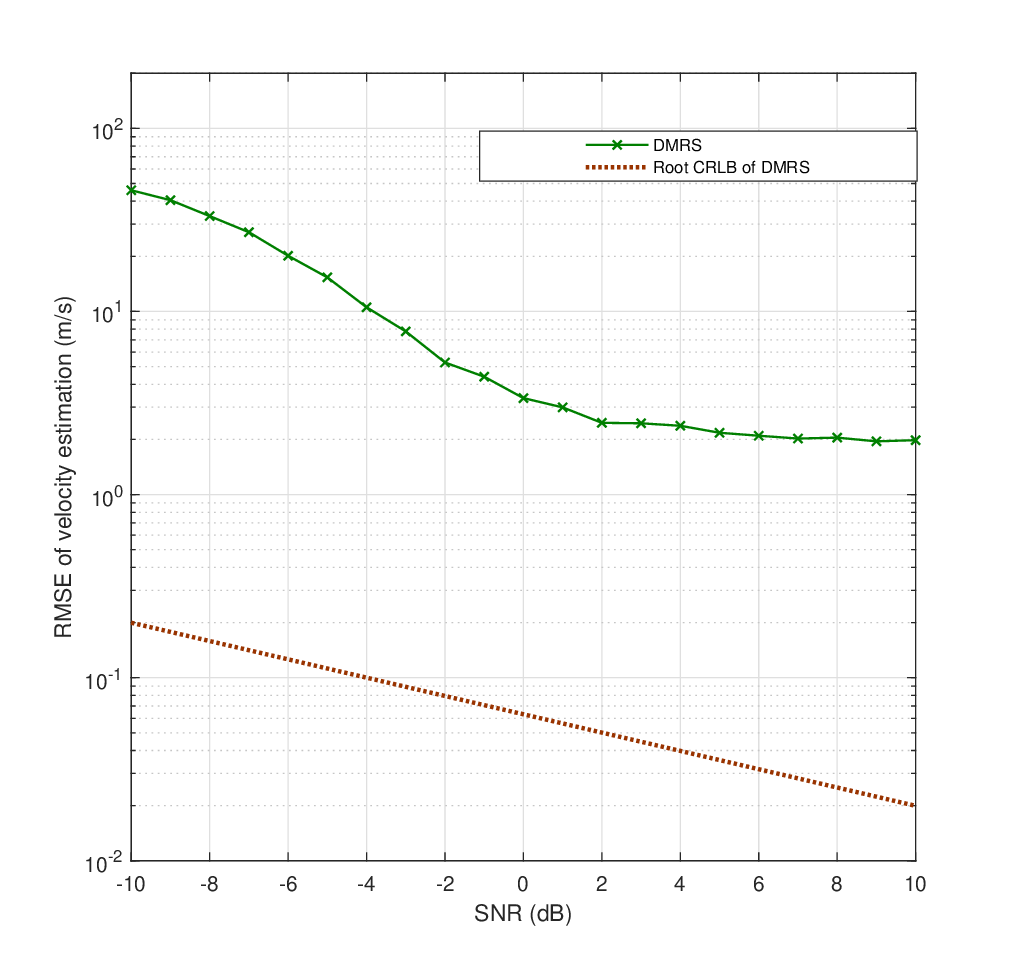}
	\caption{Comparison of RMSEs for velocity estimations among PRS, DMRS and data signal} 
	\label{fig_velocity18-RMSE-CRLB}
\end{figure}

\begin{figure}[!t]
	\centering
	\includegraphics[width=3.4in,height=3in]{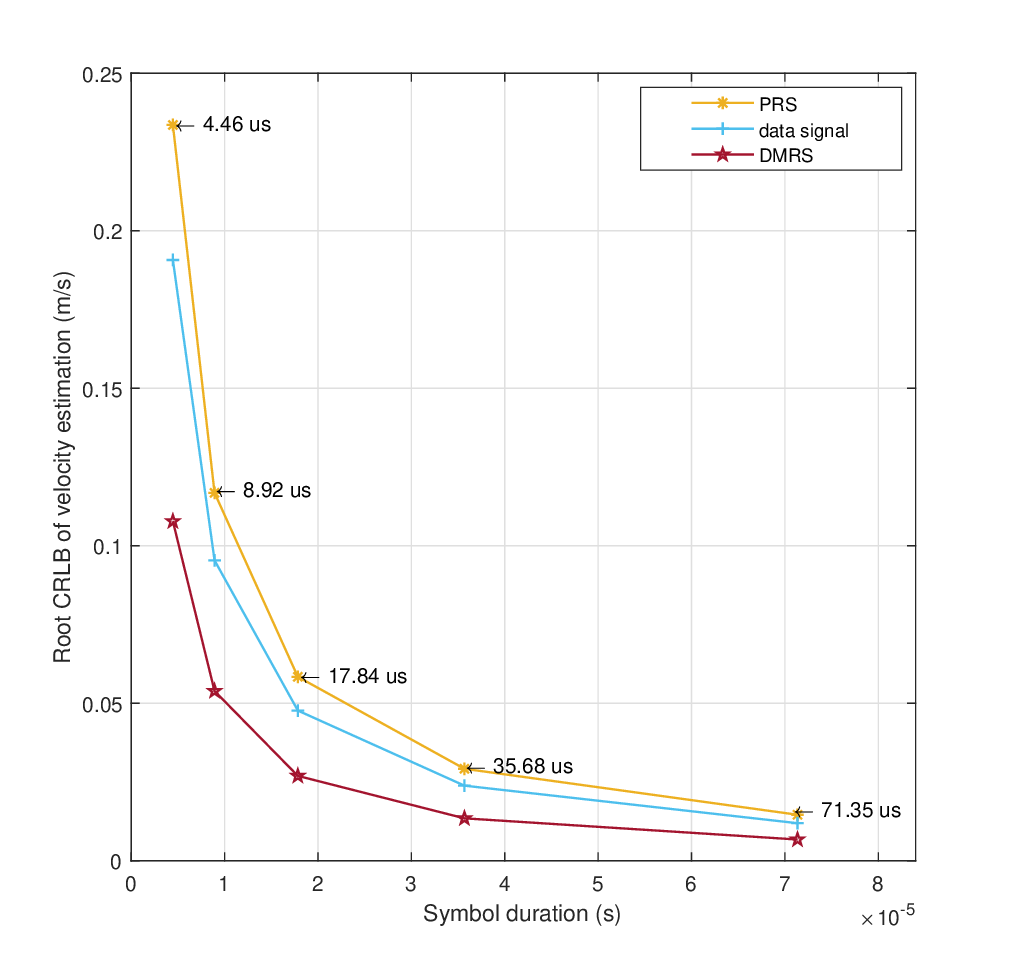}
	\caption{Root CRLB for velocity estimation under different symbol durations} \label{fig_CRLB-T}
\end{figure}

\begin{figure}[!t]
	\centering
	\includegraphics[width=3.4in,height=3in]{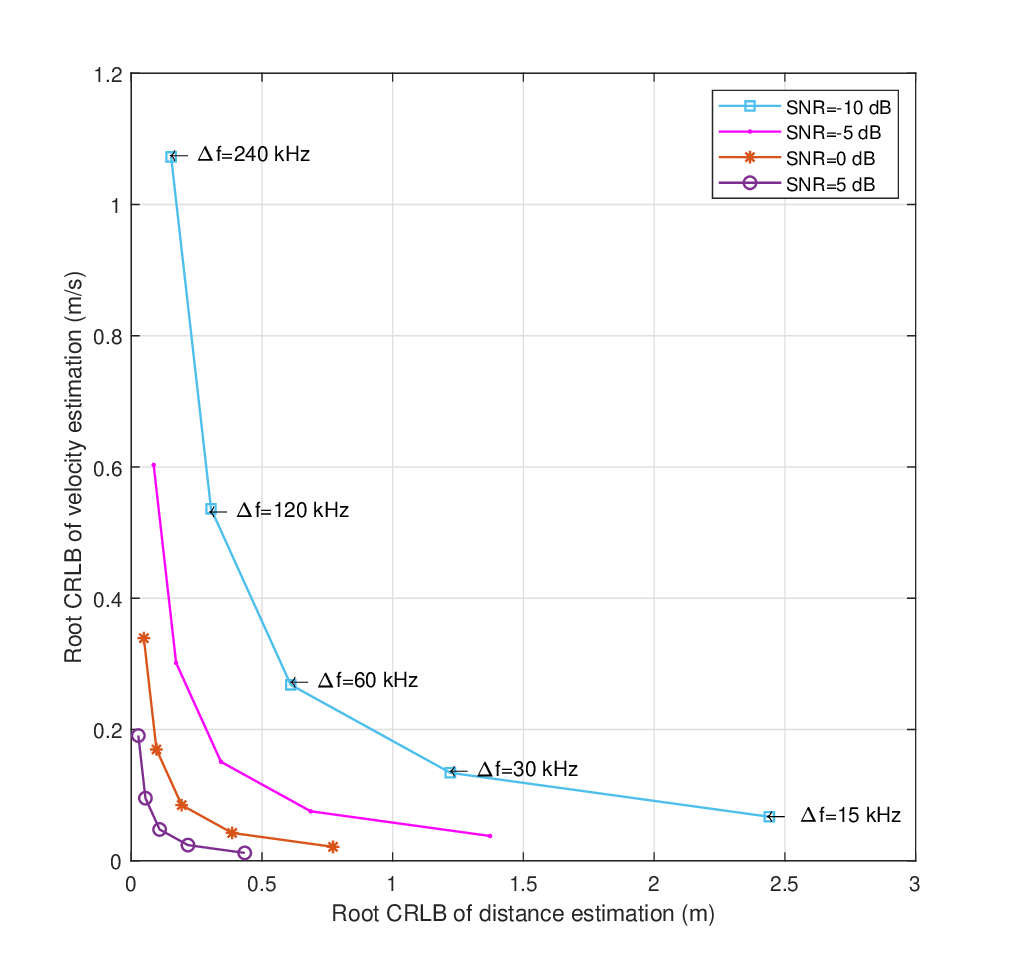}
	\caption{The relationship between root CRLB for distance and velocity estimation } \label{fig_CRLB-rv}
\end{figure}

\section{Conclusion}
\label{sec5}

This paper uses DMRS in 5G NR standard, based on typical Internet of Vehicles scenarios, conducts research on the realization of integrated sensing and communication system, and analyzes the sensing accuracy and its influencing factors in details. The simulation results show that DMRS can significantly improve the sensing accuracy when the SNR is low. There are many types of reference signals in the NR standard, occupying rich physical resources. There are many types of reference signals in 5G NR occupying abundant physical resources. Rationally arranging the time-frequency domain mapping methods of various reference signals to design ISAC signal is worthy of further exploration.

\end{document}